\documentclass[copyright,creativecommons]{eptcs}
\providecommand{\event}{PLACES 2015} 
\usepackage{breakurl}             
\usepackage{hyperref}
\usepackage{doi}

\usepackage{graphicx}
\usepackage{amsmath}
\usepackage{amssymb}
\usepackage{amsthm}
\usepackage{xspace}

\newtheorem{theorem}{Theorem}

\newtheorem{definition}{Definition}
\newtheorem{example}{Example}

\newcommand{\setof}[2]{\ensuremath{\{ #1 \mid #2 \}}}
\newcommand{\set}[1]{\ensuremath{\{ #1 \}}}

\newcommand{\ra}{\rightarrow}

\newcommand{\pow}[1]{\mathcal{P}(#1)}

\newcommand{\Nat}{\ensuremath{\mathbb{N}}}

\newcommand{\defeq}{\stackrel{\mathrm{def}}{=}}

\newcommand{\dom}[1]{\mbox{dom}(#1)}

\newcommand\lb {[\![}
\newcommand\rb{]\!]}
\newcommand{\sem}[1]{\lb #1 \rb}

\newcommand{\wbis}{\approx}

\newcommand{\newcondinfrule}[3]
           {\parbox{5.5cm}{$$ {\frac{#1}{#2}}{\qquad
            #3} \hfill  $$}}

\newcommand{\infrule}[2]
           {\parbox{4.5cm}{$$ \frac{#1}{#2}\hspace{.5cm}$$}}

\newcommand{\runa}[1]{\textsc{(#1})}

\newcommand{\Names}{\ensuremath{\mathcal{N}}}
\newcommand{\Nil}{\mathbf{0}}
\newcommand{\nil}{\mathbf{0}}

\newcommand{\para}{\mid}
\newcommand{\locate}[2][l]{#1[ #2 ]}
\newcommand{\inp}[2]{#1(#2).}

\newcommand{\binp}[2]{#1 \, #2 \,}
\newcommand{\outp}[2]{{\overline{#1}}\langle #2 \rangle.}
\newcommand{\aoutp}[2]{{\overline{#1}}\langle #2 \rangle} 
\newcommand{\bang}{!\;}

\newcommand{\new}[2][\relax]{(\nu_{#1} #2)}
\newcommand{\patt}[2]{\ensuremath{(\lambda #1 #2)}}

\newcommand{\fn}[1]{\mathrm{fn}(#1)}
\newcommand{\bn}[1]{\mathrm{bn}(#1)}

\newcommand{\near}{\triangleright}
\newcommand{\off}{\not\near\,}

\newcommand{\tabref}[1]{Table \ref{#1}}

\newenvironment{stretching}{}{}

\newcommand{\barb}[1]{\downarrow_{#1}}

\newcommand{\bbc}{BBC\xspace}

\newcommand{\conn}{\mathcal{C}\xspace}
\newcommand{\igno}{\mathcal{D}\xspace}

\newcommand{\Tenv}{\ensuremath{\Gamma}\xspace}
\newcommand{\btype}[2]{\textsf{Ch}_{\mathsf{#2}}(#1)}
\newcommand{\broad}{\mathsf{B}}
\newcommand{\collect}{\mathsf{C}}

\newcommand{\cpatt}[2]{[ \lambda #1#2 ] }
\newcommand{\mset}[1]{\{ \! | \, #1 \, | \! \}}
\newcommand{\flatten}[1]{\textsf{flat}(#1)}
\newcommand{\locs}{\ensuremath{\mathcal{L}}}

\newcommand{\wbisdot}{\ensuremath{\stackrel{\cdot}{\wbis}}}

\title{Broadcast and aggregation in \bbc}
\author{Hans H\"{u}ttel
\institute{Department of Computer Science, Aalborg University, Denmark}
\email{hans@cs.aau.dk}
\and
Nuno Pratas
\institute{Department of Electronic Systems, Aalborg University, Denmark}
\email{\quad nup@es.aau.dk
}}
\def\titlerunning{Broadcast and aggregation in \bbc}
\def\authorrunning{H. H\"{u}ttel and N. Pratas}
\begin{document}
\maketitle

\begin{abstract}

In distributed systems, where multi-party communication is essential, two communication paradigms are ever present: (1) one-to-many, commonly denoted as \emph{broadcast}; and (2) many-to-one denoted as \emph{aggregation} or \emph{collection}.

In this paper we present the \bbc process calculus, which inherently models the  \emph{broadcast} and \emph{aggregation} communication modes. We then apply this process calculus to reason on hierarchical network structure and provide examples on its expressive power. 
\end{abstract}

\section{Introduction}

In the setting of distributed systems, interprocess communication will
often be in the form of multi-party communication. Here there are two
dual paradigms: that of one-to-many communication, commonly denoted as
\emph{broadcast}, and that of many-to-one, denoted as
\emph{aggregation} or \emph{collection}.

Many-to-one communication is a commonly occurring phenomenon and often
occurs in settings with intricate network topologies. As an example
consider the case of a smart grid with thousands of meters (such as
electricity, water, heating, etc \ldots) that report their current
measurements.  A central server somewhere will collect these data;
however, the server will usually not be directly reachable from the
meters -- there will be several intermediate hops, each one facilitated
by a relay. If the underlying protocol of the network is correct, we
would expect this complicated network to be semantically equivalent to
a simple planar network, where each node is at one-hop distance from the
server.

There has already been considerable interest in understanding the
semantic foundations of one-to-many communication in the process calculus
community, starting with the work by Prasad
\cite{Prasad95} on a broadcast version of CCS. Later,
Ene and Muntean extended the notion to a broadcast version of the
$\pi$-calculus~\cite{Ene:2001:BCC:645609.662485} and proved that this
notion is strictly more expressive than
standard $\pi$-calculus synchronization.

Other process calculi with a notion of broadcast arise
in the search for behavioural models of protocols for wireless
networks. Singh et al. describe a process calculus with localities \cite{SRS:SCP10}, and Kouzapas and
Philippou \cite{kouzapas2011} introduce another process calculus with
a notion of localities, whose configuration can evolve dynamically.

In this paper we introduce a process calculus \bbc that has both forms of communication. For both many-to-one
and one-to-many communication, it is often a natural assumption that
communication is \emph{bounded}; this reflects two distinct aspects of
the limitations of a medium. In the case of broadcast, the bound
limits the number of possible recipients of a message. In the case of
collection, the bound limits the number of messages that can be
received. For this reason, \bbc uses a notion of bounded broadcast and
collection. Moreover, the syntax of the calculus introduces an
explicit notion of connectivity that makes it possible to represent a
communication topology directly. By using a proof technique introduced
by Palamidessi \cite{palamidessi} we show that even a version of \bbc
that only uses collection is more expressive than the $\pi$-calculus.

The remainder of our paper is organized as follows.  In Section
\ref{sec:bbc} we introduce the syntax of \bbc, while Section
\ref{sec:semantics} gives a reduction semantics.  In Section
\ref{sec:bisim} we introduce a notion of barbed bisimilarity, and in
Section \ref{sec:case} we use this notion to prove the correctness of
a protocol that uses both broadcast and collection. In Section
\ref{sec:typesystem} we outline a simple type system for \bbc in which
channels can be distinguished as being used for broadcast or for
collection. Finally, in Section \ref{sec:expressiveness} we show that
even a version of \bbc that only uses collection is more expressive
than the $\pi$-calculus.

\section{The syntax of \bbc} \label{sec:bbc}


\subsection{The syntactic categories}

A central notion in \bbc is that of \emph{names}. As in the
distributed $\pi$-calculus \cite{DBLP:conf/popl/RielyH98}, processes
reside at named sites, called \emph{locations}, and use named
\emph{channels} for communication.  We assume that names are taken
from a countably infinite set $\textbf{Names}$. In general, we denote
names of channels by $a,b,c \ldots$, names of locations by
$l,m,n \ldots$ and if nothing is assumed about the usage of the name
we denote them by $x,y,z \ldots$.
 
We let $M \in \textbf{Msg}$ range over the set of messages, let $P$ range over the set
of processes and let $N$ range over the set of networks. Since a
collecting input (defined below) can receive a multiset of messages,
each coming from a distinct sender, we also consider \emph{multiset
  expressions} $E$ and multiset variables $S$ that can be instantiated
to multiset expressions. The formation rules defining the syntactic
categories of \bbc are given below. 

\begin{table}
       \begin{align*}
              M  & ::= x \mid S \mid (M_1,\ldots, M_k) \mid g(E) \mid f(M) \\[2mm]
              P  & ::= \binp{a}{\patt{\vec{x}M}}P_1 \mid
              \binp{a}{\cpatt{\vec{x}M.S}} P_1 \mid \outp{a}{M}P_1   \mid \new{x:\beta_x}P_1 \mid[M_1=M_2]P_1 \mid [M_1 \neq M_2]P_1 \mid \left(P_1 \para P_2\right)  \\[2mm]
              & \mid \left(P_1 + P_2\right) \mid \Nil \mid A \left(\vec{M}\right) \\[2mm]
              E & ::= \{ M_1, \ldots, M_k \} \mid S \\[2mm]
              N &  ::= \locate{P} \mid \left(N_1 \para N_2\right) \para \new{x:b}N \para l \near m \\
       \end{align*}
       \caption{Formation rules for \bbc}
       \label{tab:FormationRules}
\end{table}

\subsection{Messages and patterns}

For ordinary expressions we assume a collection of term constructors
ranged over by $f$, that build messages out of other
messages. Moreover, we assume the existence of a collection of
multiset selectors ranged over by $g$; these can be used to build
messages out of multisets.

If a channel is to be chosen among a collection of candidate channels, we can use a multiset selector to describe this.
\begin{example}
  An example of a multiset selector is the function
  $\text{find-a}$ that intuitively returns the name $a$ if this name
  occurs as the first component in a multiset of pairs of names and
  the name $k \neq a$ and is defined by $\mbox{find-a}(S) = a \;
  \text{if}\; (a,x) \in S \;\text{for some } x$ and $k$ otherwise.
\end{example}

Practical examples of interest are the election of a common channel in
a ad-hoc network; selection of a channel for cooperative sensing or for communication within an ad-hoc cluster. 

An important notion is that of an \emph{input pattern} which is of the form $\patt{\vec{x}M}$, where the variable names in $\vec{x}$ are distinct and occur free in $M$.
A message $O$ matches this pattern, if it can be obtained from it through substitution.

More formally, a \emph{term substitution} is a finite function
$\theta: \textbf{Names} \ra \textbf{Msg}$.  The substitution can also
be written as a list of bindings $\theta = [x_1 \mapsto M_1,\ldots,
x_k \mapsto M_k]$. The action of $\theta$ on an arbitrary message or
multiset expression is defined in the expected way.  $M'$ is said to
match $\patt{\vec{x}}{M}$ with $\theta$ if for a substitution $\theta$
with $\dom{\theta} = \vec{x}$ $M' = M\theta$ is true.

\subsection{Processes}

In a collecting communication setting, the receiver can make no assumption
about the number of messages that will be received, nor on the order
in which they are received. Moreover, we cannot assume that a message
that has arrived will only occur once among the messages received
during a single collecting communication. We shall therefore think of
a collecting input as receiving a multiset of messages.

There are two kinds of input prefixes in \bbc:
\begin{itemize}
\item The \emph{broadcast input} $\binp{a}{\patt{\vec{x}M}}P_1$ in
  which a single term matching the pattern $\patt{\vec{x}M}$ is
  received on the channel $a$. The pattern variables in $\vec{x}$ are
  bound in $P_1$ and get instantiated with the appropriate subterms
  that correspond to the pattern.
\item The \emph{collection input} $\binp{a}{\cpatt{\vec{x}M.S}} P_1$
  in which a non-empty multiset of terms $\{ M_1, \ldots, M_K \}$ each of which matches the pattern
  $\cpatt{\vec{x}}{M}$ is received on the channel $a$. Note that in this
    case the scope of the pattern variables in $\vec{x}$ \emph{does
      not extend} to $P_1$. Following the input, the multiset variable $S$
    is instantiated to the multiset $\{ M_1, \ldots, M_K \}$.
\end{itemize}

In a \emph{restriction} $\new{x:\beta_x}P_1$ the notion of bounded
communication is made explicit: we declare the name $x$ to be private
within $P_1$ and to have bound $\beta_x$, where $\beta_x :
\textbf{Names} \ra \Nat$ is a function such that for any location name
$m$ we have that $\beta_x(m) = k$, if it is the case that for a
process located at $m$ there are at most $k$ senders that are able to
send a message to it using the channel $x$.

The remaining process constructs are standard.
 The \emph{output}
process $\outp{a}{M}P_1$ sends out the message $M$ on the channel
named $a$ and then continues as $P_1$.  \emph{Match}
$[M_1=M_2]P_1$ and \emph{mismatch} $[M_1 \neq M_2]P_1$ proceed as
$P_1$ if $M_1$ and $M_2$ are equal, respectively
distinct. \emph{Parallel composition} $P_1 \para P_2$ runs the
components $P_1$ and $P_2$ in parallel. \emph{Nondeterministic choice}
$P_1 + P_2$ can proceed as either $P_1$ or $P_2$; \emph{inaction}
$\Nil$, has no behaviour.

Finally, we allow agent identifiers $A(\vec{M})$ parameterized by a sequence of
messages; an identifier must be defined using an equation of the
form $A(\vec{x}) \defeq P$. The only names free in $P$ must be the
parameters found in $\vec{x}$, that is $\fn{P} \subseteq \vec{x}$.
Definitions of this form can be recursive,
with occurrences of $A(\vec{x})$ (with names in $\vec{x}$
instantiated by concrete messages) occurring within $P$.

Restriction and broadcast input are name binders; for a process $P$, the sets of \emph{free
  names} $\fn{P}$ and \emph{bound names} $\bn{P}$ of $P$ are defined
as expected. For collection input we define
\[ \bn{\binp{a}{\cpatt{\vec{x}{M\, S}}}P_1} =  \bn{P_1}  \]

Replication, denoted as $\bang P$, is a derived construct in \bbc; a
replicated process $\bang P$ is expressed by the agent identifier
$A_P$ whose defining equation is $A_P = P \para A_P$ and should
therefore be thought of as an unbounded supply of parallel copies of
$P$.

\subsection{Networks}

As in the 
distributed
$\pi$-calculus~\cite{DBLP:conf/popl/RielyH98} a parallel composition
of located processes is called a \emph{network}.  According to
the formation rules, a \emph{network} $N$ can be a process
$P$ running at location $l$, which is denoted as $\locate{P}$. We also
allow parallel composition $P_1 \para P_2$ and restriction $\new{x}N$
at network level. Moreover, the \emph{neighbourhood predicate} $l
\near k$ denotes that location $l$ is close to $k$. For the
neighbourhood predicate, parallel composition is thought of as logical
conjunction. So, if $l$ and $k$ are close to each other, then
$l\bowtie m$ can be written instead of $l \near m | m \near l$.

          For any term $M$, the set of free names $\fn{M}$ is defined
          in the standard way. 

The usual notions of $\alpha$-conversion also apply here. We write
$P_1\equiv_{\alpha} P_2$ if $P_1$ can be obtained from $P_2$ by
renaming bound names in $P_2$, and likewise we write
$N_1\equiv_{\alpha} N_2$ if $N_1$ can be obtained from $N_2$ by
renaming bound names in $N_2$.

\section{The semantics of \bbc} \label{sec:semantics}

We now describe a reduction semantics of \bbc.

\subsection{Evaluation of message terms}

In our semantics, we rely on an evaluation relation $\leadsto$ defined
for both message terms and multiset expressions. 

\begin{example} Assume that our set of function symbols contains the
  projection function which extracts the first coordinate of a pair of
  names and that this is the only function symbol that can lead to
  evaluation. The evaluation relation can then be defined by the axiom
\[ \text{first}(x,y) \leadsto x \]
and the rules
\begin{align*}
       \frac{M \leadsto M'}{\mset{M,\ldots} \leadsto \mset{M',\ldots}} \quad
       \frac{M \leadsto M'}{f(M) \leadsto f(M')} \quad
       \frac{E \leadsto E'}{g(E) \leadsto g(E')}
\end{align*}
\end{example}

A message term $M$ is \emph{normal} if $M \not\leadsto M_1$ for any
$M_1$.

\subsection{Structural congruence and normal form}
\label{sec:StructuralCongruenceAndNormalForm}

Structural congruence $\equiv$ is defined for both processes and networks; two processes (or networks) are related, if they are identical up to simple structural modifications such as the ordering of parallel components.
The relation is defined as the least equivalence relation satisfying the proof rules and axioms of Tables~\ref{table:scong-proc} and \ref{table:scong-net}.
\begin{table}[t]
                     \begin{tabular}{clcl}
                            \runa{P-Out}                    &
                                $\infrule{M \leadsto M'}{\outp{a}{M}P
                                  \equiv \outp{a}{M'}P}$ & 
                            \runa{P-Com}                & $P_1 \para P_2 \equiv P_2 \para P_1$  \\[1mm]
                            \runa{P-As}                 & $(P_1 \para P_2) \para P_3 \equiv P_1 \para (P_2 \para P_3)$ &
                            \runa{P-Com-Plus} & $P_1 + P_2 \equiv P_2 + P_1$ \\[1mm]
                            \runa{P-As-Plus}     & $(P_1 + P_2) + P_3 \equiv P_1 + (P_2 + P_3)$ &
                            \runa{P-Nil}                & $P \para \Nil \equiv P$ \\[1mm]
                            \runa{P-Ext}                & $\new{x}(P_1 \para P_2) \equiv \new{x}P_1 \para P_2$ if $x \not\in \mathrm{fn}(P_2)$ &
                            \runa{P-New}                & $\new{x}\new{y} P \equiv \new{y}\new{x} P$ \\[1mm]
                            \runa{P-Eq-1}               & $[M = M]P \equiv P$ &
                            \runa{P-Eq-2}               & $[M \neq N]P    \equiv P$ \\[1mm]
                            \runa{P-Eq-3}               & $\infrule{M
                                  \leadsto M'}{[M \neq N]P
                                \equiv [M' = N]P}$ &
                            \runa{P-Eq-4}               & $\infrule{M
                                  \leadsto M'}{[M = N]P
                                \equiv [M' = N]P}$ \\[3mm]
                                \runa{P-Eq-5}                  &  $[M
                                \neq N] P \equiv [N \neq M]P$&
                             \runa{P-Eq-6}                  &  $[M
                                = N] P \equiv [N = M]P$ \\[3mm]
                            \runa{P-Alpha}              &  $\infrule{P_1 \equiv_{\alpha} P_2}{P_1\equiv P_2}$  &
                            \runa{P-AG}                 & $ \infrule{A(\vec{x}) \defeq P}{A(\vec{M}) \equiv P\theta}$ \\
                                \runa{P-Par}    &   $\infrule{P_1
                                  \equiv P_2}{P_1 \para Q \equiv P_2
                                  \para Q}$ &
                                \runa{P-Res}    &   $\infrule{P_1
                                  \equiv P_2}{\new{x} P_1 \equiv
                                  \new{x} P_2 }$ \\[1mm]
                     \end{tabular}
       \caption{Structural congruence for processes}
       \label{table:scong-proc}
\end{table}
In the rules defining structural congruence, the parallel composition is commutative and associative both at the level of processes (rules \runa{P-Com} and \runa{P-As}) and of networks (rules \runa{N-Com} and \runa{N-As}).
This justifies the use of iterated parallel composition $\prod_{i \in I} P_i$ introduced earlier.

The restriction axioms describe the scope rules of name restrictions.
The scope extension axioms \runa{P-Ext} and \runa{N-Ext} express that the scope of an extension can be safely extended to cover another parallel component if the restricted name does not appear free in this component.
The exchange axioms \runa{P-New} and \runa{N-New} allow us to exchange the order of restrictions. 

The agent axiom \runa{P-Ag} tells us that an instantiated agent
identifier $A(\vec{n})$ should be seen as the same as the right-hand
side of its defining equation $A(\vec{x}) = P$ instantiated by the
substitution $\theta$ that maps the names in $\vec{x}$ component-wise
to those of $\vec{n}$. 

Finally, the rules \runa{P-Par} and \runa{P-Res} express that
structural congruence is indeed a congruence relation for parallel
composition and restriction.
\begin{table}[t]
              \begin{stretching}
                     \begin{tabular}{c lcl}
                            \runa{N-Cong} & \infrule{P_1 \equiv P_2}{\locate{P_1} \equiv \locate{P_2}} &
                            \runa{N-Com}  & $N_1 \para N_2 \equiv N_2 \para N_1$ \\
                            \runa{N-As}   & $(N_1 \para N_2) \para N_3 \equiv N_1 \para (N_2 \para N_3)$ &
                            \runa{N-Nil}  & $N \para \Nil \equiv N$ \\
                            \runa{N-Ext}  & $\new{x}(N_1 \para N_2) \equiv \new{x}N_1 \para N_2$  if $x \not\in \mathrm{fn}(P_2)$ &
                            \runa{N-New}  & $\new{x}\new{y} N \equiv \new{y}\new{x} N$ \\
                            \runa{N-Loc}  & $\locate{P \para Q} \equiv \locate{P} \para \locate{Q}$ &
                            \runa{N-Eq}   & $\new{n}\locate{P}  \equiv \locate{\new{n}P} \quad \text{if } l \neq n$ \\[3mm]
                            \runa{N-Alpha} &  $\infrule{N_1
                                  \equiv_{\alpha} N_2}{P_1\equiv N_2}$
                                &
                               \runa{N-Par}    &   $\infrule{N_1
                                  \equiv N_2}{P_1 \para N_3 \equiv N_2
                                  \equiv N_3}$ \\[1mm]
                                \runa{N-Res}    &   $\infrule{N_1
                                  \equiv N_2}{\new{x} N_1 \equiv
                                  \new{x} N_2 }$ \\[1mm]

                     \end{tabular}
              \end{stretching}
       \caption{Structural congruence for networks}
       \label{table:scong-net}
\end{table}

By using the laws of structural congruence, any network can be
rewritten to normal form.  Informally, a network is in normal form if
it consists of the total neighbourhood information as one parallel
component and the location information as the other.  In the
following, the notation $\prod_{i \in I} P_i$ is used to denote
iterated parallel composition, i.e. the parallel composition of $P_1,
\ldots, P_k$, where $I = \{ 1, \ldots, k \}$.

\begin{definition}
A network $N$ is on \emph{normal form} if 
\(  N = \new{\vec{m}} ( \conn \para
\prod_{k \in K} \locate[l_k]{P_k} ) \)
where 
\(
\conn  = \prod_{i \in I} \prod_{j \in I \setminus \set{i}} l_i \near m_j .
\)
\end{definition}
We write $l \near m \in \conn$ if $\conn \equiv l \near m \para
\conn'$ for some $\conn'$.
It is easy do see that any network $N$ can be rewritten into normal
form.
\begin{theorem}\label{thm:nf}
  For any network $N$, there exists a network $N_1$ such that $N_1 \equiv N$ and $N_1$ is on normal form.
\end{theorem}

\begin{proof}
Induction in the structure of $N$. The proof is similar to that for
the $\pi$-calculus \cite{SanWalk}; the idea is to use the scope
extension axioms to push out restrictions while $\alpha$-converting bound
names whenever needed.
\end{proof}

\subsection{The reduction relation}
\begin{table*}
              \begin{tabular}{c l}
                     \runa{Broad} & $\begin{array}{ll}
                                          \new{\vec{n}:\vec{\beta}}(\conn \para \locate[l]{\outp{a}{M}P_k} \para \prod_{i=1}^{k} \locate[m_i]{\binp{a}{\patt{\vec{x}}{M'_i}}Q_i} \para N_1 )&\\
                                          \longrightarrow & \\
                                          \new{\vec{n}:\vec{\beta}}(\conn  \para \locate[l]{P_k } \para
                                          \prod_{i=k'+1}^{k}\locate[m_i]{\binp{a}{\patt{\vec{x}}{M'_i}}Q_i} \para \prod_{i=1}^{k'}\locate[m_{\ell}]{Q_{i}\theta_{i}} \para N_1)&                                   
                                          \end{array}$ \\
                                          & where $l \near m_i \in \conn$ for all $1 \leq i \leq m$ and either\\
                                          &$\quad\quad$ (1) $a \not\in \vec{n}$, $k' =|\setof{m_{\ell}}{l_k\near m_{\ell} \in \conn}|\leq b(a,l)$ for all $\ell \in L$, or \\
                                          &$\quad\quad$ (2) $a \in \vec{n}$, $k' \leq |\setof{m_{\ell}}{l_k \near   m_{\ell}  \in \conn}| \leq \beta_a(l)$ for all $\ell \in L$,\\
                                         
                                          & and for all $1 \leq i \leq m'$ we have $M = M'_{i}\theta_{i}$
                                          for some $\theta_{i}$,\\
                     \runa{Local} &     $\begin{array}{ll}
                                                 \new{\vec{n}:\vec{\beta}}(\conn \para\igno \para\locate[l]{\outp{a}{M}P \para \binp{a}{\patt{\vec{x}}{M'}}Q \para P'}
                                                 \para N_1) &\\ 
                                                 \longrightarrow & \\
                                                 \new{\vec{n}:\vec{\beta}}(\conn_1 \para \igno_1 \para  \locate[l_k]{P \para Q\theta \para P' } \para N_1) \\
                                          \end{array}$\\[2mm]
                                    &     where $M = M'\theta$ for some $\theta$ \\[6mm]
 
                     \runa{Coll} &   $\begin{array}{ll}
                                                 \new{\vec{n}:\vec{\beta}}(\conn \para \prod_{k \in K} 
                                                 \locate[l_k]{\outp{a}{M_k}P_k \para P'_k} \para 
                                                 \locate[m]{\binp{a}{\cpatt{\vec{x}M.S}}Q_\ell} \para N_1 ) \\
                                                 \longrightarrow & \\
                                                 \new{\vec{n}:\vec{\beta}}(\conn \para \prod_{k \in K }
                                                 \locate[l_k]{P_k \para P'_k} \para 
                                                 \locate[m]{Q_{\ell}[S\mapsto\mset{M_1,\ldots,M_{|K|}} } \para N_1)
                                          \end{array}$ \\
                                   &      where $l_j \near m \in \conn$ for all $j \in K$ and either
                                          \begin{tabular}{l} 
                                                 $a \not\in \vec{n}$, and $1 \leq |K| \leq b(a,m)$ or \\
                                                 $a \in \vec{n}$ and $1 \leq |K| \leq \beta_a(m)$
                                          \end{tabular}\\
                                   & and for all $j \in K$ we have $M_j = M\theta_{\ell}$ for some $\theta_{\ell}$\\
                     \end{tabular}
       \caption{Reduction rules for communication in \bbc networks on normal form, assuming connectivity $b(a,m)$}
       \label{table:redsem2}
\end{table*}

In our reduction rules, we assume that network terms are on normal
form, as defined above. The rules are given
in Table \ref{table:redsem2}.

We call a location that contains an available input a \emph{receiving
  location} for channel $a$ and a location which contains an availabke
output a \emph{sending location}. For broadcast, we require that if
the number of receiving locations for a channel $a$ exceeds the bound
$b(a,l)$ for $a$ (or $\beta_a(m)$, if $a$ is a bound name), then at
most $b(a,l)$ receiving locations can receive the message. All other
receiving locations do not receive anything and will still be waiting
for an input. Each of the receiving locations must be connected to the
sending location $l$. This is captured by the rule \runa{R-Broadcast}.

For collection, the number of locations that can simultaneously send
messages on a channel $a$ cannot exceeds the bound $b(a,l)$ for $a$
(or $\beta_a(m)$, if $a$ is a bound name). Moreover, each sending
location must be connected to the receiving location $m$. This is
captured by the rule \runa{R-Collect}. Finally, the reduction rule
\runa{R-Local} describes local communication within the confines of a
single location.

\begin{example}
Consider the network
\begin{align*}
       N & = l_1 \near l_3 \para l_2 \para l_3 \para \locate[l_1]{\aoutp{a}{(a,b)}} \\
          &  \para \locate[l_2]{\aoutp{a}{(c,b)}} \para \locate[l_3]{\binp{a}{\cpatt{x}{(x,b).S}}.\aoutp{d}{\text{find-a}(S)} }
\end{align*}

Assume that $b(a) = 2$. This network consists of three locations. Location $l_1$ offers an output on the $a$-channel of the pair $(a,b)$, and location $l_2$ offers
an output on the $a$-channel of the pair $(c,b)$. At location $l_3$ we
have a process that on the channel $a$ will receive a set of messages,
all of which are of the form $(x,b)$ for some $x$, and subsequently
output $a$ if one of the pairs received contained $a$.

We have $ N \ra  l_1 \near l_3 \para l_2 \para l_3 \para
\locate[l_1]{\nil} \para \locate[l_2]{\nil} \para
\locate[l_3]{\aoutp{d}{a}}$.
\end{example}

\section{Bisimilarity in \bbc} \label{sec:bisim}

\subsection{Barbs}

In our treatment of bisimilarity, we define an observability predicate
(aka barbs) We write $P \barb{a,\broad}$ if $P$ admits a broadcast
observation on channel $a$, and $\vdash P \barb{a,\collect}$ if
$P$ admits a collection observation on channel $a$. We can define
a similar predicate for networks. We write $N \barb{a,d}@l$ if $N$
allows a barb $\barb{a,d}$ at location
$l$. \tabref{tab:barbs-networks} contains the most interesting rules.
\begin{table}[h]
  \begin{center}
  \begin{tabular}{llll}
  \runa{Inp-B} & $ \binp{a}{\patt{\vec{x}M}}P_1 \barb{a,\broad}$ & 
  \runa{Inp-C} & $ \binp{a}{\cpatt{\vec{x}M.S}}P_1 \barb{a,\collect}$ \\[4mm]
  \runa{Outp} &  $\outp{a}{M}P \barb{a,d}$ & 
  \runa{Par} & \infrule{ P_1 \barb{a,d}}{
  P_1 \para P_2 \barb{a,d}} \\[4mm]
 \runa{New} & \newcondinfrule{ P_1 \barb{a,d}}{
  \new{x :
    \beta_n}P_1 \barb{a,d}}{x \neq a} & 
    \runa{Locate} & \infrule{ P \barb{a,d}}{
      \locate{P} \barb{a,d}@\ell} \\
\runa{Cong} & \infrule{ N_1 \barb{a,d}@l \quad N_1 \equiv
  N}{ N \barb{a,d}@l} 
  \end{tabular}
\end{center}
  \caption{Selected rules for barbs}
  \label{tab:barbs-networks}
\end{table}

\begin{definition}
A weak barbed bisimulation $R$ is a symmetric binary relation on
networks which satisfies that whenever $N_1 R N_2$ we have
\begin{enumerate}
\item If $N_1 \ra N'_1$ then for some $N'_2$ we have $N_2 \ra^* N'_2$
  where $N'_1 R N'_2$ 
\item For every location $l$, if $N_1  \barb{a,d}@l$ then also $N_2  \barb{a,d}@l$
\end{enumerate}
We write $N_1 \wbisdot N_2$ if $N_1 R N_2$ for some weak barbed
bisimulation $R$.
\end{definition}

\begin{theorem}
Weak barbed bisimilarity is an equivalence.
\end{theorem}

As in other process calculi, the notion of barbed bisimilarity is
unfortunately not a congruence, as it is not preserved under parallel
composition.

Two networks are weak barbed congruent if they are barbed bisimilar
in every parallel context, i.e. no surrounding network can tell them
apart.

\begin{definition}
We write $N_1 \wbis N_2$ if for all networks $N$ we have that 
\( N_1 \para N \wbisdot N_2 \para N \)
\end{definition}

\section{A hierarchical protocol}
\label{sec:case}

In this section we outline how one can reason about a distributed
protocol that involves both collection and broadcast. The protocol
itself is generic but representative of the current trend in
distributed communication systems with a large number of devices and
very few controlling/central entities.  In essence, this protocol is
composed by traffic collection in the upstream
direction, i.e. from the leaves to the central entity, and traffic
broadcast in the downstream direction, i.e. from the central entity to
the leaf nodes.
\begin{figure}
       \centering
              \includegraphics[width=\linewidth]{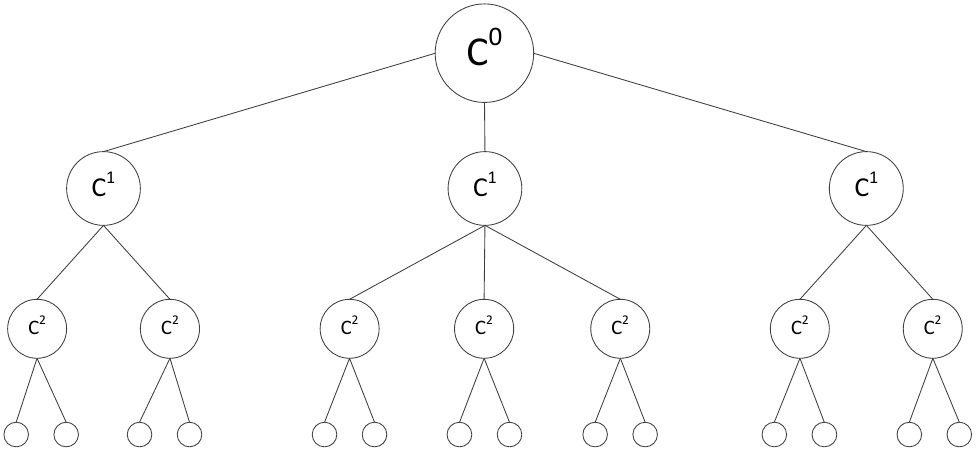}
       \caption{A multi-level network.}
       \label{fig:MultiLevelNetwork}
\end{figure}

Figure~\ref{fig:MultiLevelNetwork} shows the multi-level network
topology assumed for this protocol; the topology is that of a tree.
Each level of the network is connected to the one above through a local
collection and broadcast node, which we simply denote as local central node.
This node is then a leaf node of the above level, e.g. the
nodes denoted as $C^1$ are leaf nodes of the level $l_0$, while being
the local central nodes at each of the locations of $l_1^i$, which for
ease of notation we denote solely as $l_1$. We assume that there can
be at most $\beta$ nodes connected to every central node, meaning that the
bound of every channel should be $\beta$.

The goal of the distributed protocol is to collect information from
all leaf nodes in the network; and then communicate it to a central entity,
denoted as $C^0$.  The central entity $C^0$, then reaches a global
decision, which is then communicated to the leaf nodes. For simplicity, we assume that
this decision only depends on the data collected from the leaves.

\subsection{The $n$-level protocol in \bbc}
\label{sec:TheProtocolInBbc}

In what follows, we define the components of the network and protocol inductively.
We assume the network to have $n$ levels or depth $n$, where level $0$ corresponds to the central entity, while level $n$ denotes the level where all the leaf nodes are.

One of the sub-networks of the $n$-level protocol at level $k$ , as
depicted in Figure \ref{fig:MultiLevelNetwork}, is composed by $m$
nodes and a local central node $C_{k}$.  We denote this sub-network by
the agent identifier
$D^{k,n}\left(a'_k,a''_k,b'_k,b''_k,\locs,\ell_{k-1}\right)$ where
$\locs$ is a set of locations and $\ell_{k-1}$ is the local central
location at $k-1$ level. The other name parameters are
\begin{itemize}
       \item $a'_k,a''_k$ - names of the channels used for
         communications at location $l_k$, i.e. between the leaf nodes
         and the local central entity; the names transmitted are collected by the central node at
         level $k+1$. 
       \item $b'_k,b''_k$ - names of the channels that the local
         central entity uses to communicate to the local central
         entity for which it is a lead, i.e. the channel over which
         the local central entity $C_k$ communicates with the local
         central entity $C_{k-1}$. The names transmitted are broadcast to the child nodes of
         the local central (at level $k-1$). 
       \item $\locs$ is the set of nodes at this level.
\end{itemize}

At the leaf level, $k=0$ we have,
\begin{align*}
       D^{0,n}(a'_0,a''_0,b'_0,b''_0\locs,\ell_{0}) & \defeq \prod_{\ell_i \in \locs}
        \locate[\ell_i]{P_i(a'_0,a''_0)} \para \prod_{\ell_i \in \locs} \ell_i \near \ell_{1} \para \locate[\ell_{0}]{C^0(a'_0,a''_0,b'_0,b''_0)}
\end{align*}
The subprocess $P_i$ where $\ell_i \in \locs$ is defined by
\begin{equation*}
  \label{eq:1}
P_i(a'_0,a''_0) \defeq \new{\omega_i} \outp{a'_0}{\omega_i}\inp{a''_0}{z} (Q_i
(z)\para  P_i(a'_0,a''_0)  )
\end{equation*}
where $Q_i(z)$ depicts the computation that occurs at the leaf
$\ell_i$; we shall not describe this here and $\omega_i$ is the local
name that is the contribution made by the process.

The local central node is defined by
\begin{align*}
  \label{eq:2}
  C^{0,n}(a',a'',b',b'') \defeq \inp{a'}{S}\outp{b'}{f(S)}\inp{b''}{z}\outp{a''}{z}
   C^{0,n}(a',a'',b',b'')
\end{align*}
This process will receive names that are collected to the set $S$ and
then pass the processed information $f(S)$ upwards on the $b$ channel.
After this, the local central node waits for a response to the
information sent and passes on the received name to its children.

The \emph{selection function} $f: \pow{\Names} \ra \Names$ processes
the messages received from the leaf nodes and selects a name; this
could be e.g. a channel estimate. We require that selection is
\emph{idempotent} in the sense that for any family of multisets $S_1,
\ldots, S_k$ we have
\[  f(\set{f(S_1),\ldots,f(S_k)}) = f(\cup^{k}_{i=1} S_i) \]
If we think of names as natural numbers, the function
$f(S) = \min_{x \in S} x$ is an example of an idempotent selection function.

For intermediate levels, $0 < k < n$, $D^k\left(a_k,b_k,\ell_k,\ell_{k-1}\right)$ is defined as follows,
\begin{align*}
  \label{eq:3}
   D^{k+1,n}(a'_{k+1},a''_{k+1},b'_{k+1},b''_{k+1},\locs, \ell_{k+1}) \defeq 
  \new{\vec{\ell}, a'_{k} : \beta,a''_{k} : \beta,b'_{k}
    : \beta,b''_{k} : \beta} \\ \left(\prod_{\locs_i \in \mathsf{Locs}(\locs,k)} D^{k,n}(a'_{k},a''_k,b'_{k},b''_k,\locs_i,\ell_{i})
\para \right. \left. \prod^{m}_{i=1} \ell_i \near \ell_k \para \locate[\ell_{k+1}]{C^{k+1}(a'_k,a''_k,b'_k,b''_k)} \right)
\end{align*}
where the set of locations $\mathsf{Locs}(\locs,k)$ can be partitioned
as $\mathsf{Locs}(\locs,k) = \bigcup_i{i \in I}\locs_i$ where $I$ is
some index set and $\locs_{i} \cap \locs_{j} = \emptyset$ for all $i
\neq j, i,j \in I$.  $m'_i$ is the number of neighbours in the $i$'th
sub-network found at level $k+1$ and $\ell_i \in \vec{\ell}$.

The definition of local central nodes is (for $k < n$)
\begin{align*}
  C^{k,n}((a_k,b_k) \defeq 
  \inp{a_k}{S}\outp{b_k}{f(S)}\inp{b_k}{z}\outp{a_k}{z} (Q_c
  (z) \! \para\!  C^k(a_k,b_k))
\end{align*}

Finally, at level $n$ the information collected from all the leaf nodes reaches the central entity, where the final processing of the received information is performed.
The description of the entire network $D^n$ is $D^{n,n}_m(a_0,\ell_c)$, defined as
\begin{align*}
 D^{n,n}(a'_n,a''_n,b'_n,b''_n,\ell_n) \defeq \new{\vec{\ell},
   a'_{n-1}: \beta,a''_{n-1}: \beta,b'_{n-1}: \beta,b''_{n-1}: \beta} \\
   \left(\prod_{\locs_i \in \mathsf{Locs}(\locs,n)} D^{1,n}_{m'_i}(a'_{n-1},a''_{n-1},b'_{n-1},b''_{n-1},\locs_i,\ell_{n}) \right.\para 
& \left. \prod^{m}_{i=1} \ell_i \near \ell_0 \para \locate[\ell_0]{C^n(a'_n,a''_n,b'_n,b''_n)} \right)
\end{align*}
Here the central entity will collect the estimates from its children
and then broadcast back the value of the selection.
\begin{equation}
  \label{eq:5}
  C^{n,n}(a'_n,a''_n,b'_n,b''_n) \defeq \inp{a'_n}{S}\outp{b'_n}{f(S)} C^{n,n}(a'_n,a''_n,b'_n,b''_n))
\end{equation}

\subsection{Correctness}
\label{sec:Correctness}

We would like the multi-level network
that serves $K$ leaf nodes to be equivalent to a single-level network
containing the same $K$ leaf nodes but now connected to a single
central entity.

To this end we define the flattened version of a hierarchical network $N$,
denoted $\flatten{N,\ell}$. This is precisely a network in which all nodes
are connected to the same central node at location $\ell$.

This operation is defined inductively as follows.
\begin{align*}
 \flatten{D^{0,n}(a'_0,a''_0,b'_0,b''_0,\locs,\ell_0),\ell} \defeq 
 D^{0,n}(a'_0,a''_0,b'_0,b''_0,\locs,\ell_0)  \\[2mm]
\flatten{D^{k+1,n}(a'_{k+1},a''_{k+1},b'_{k+1},b''_{k+1},\locs,\ell_{k+1}),\ell}
\defeq  \\ \new{\vec{\ell},
    a'_{k},a''_{k},b'_k,b''_k} 
\left( \prod^{m}_{\locs_i \in \mathsf{Locs}(\locs,k)}
    \flatten{D^{k,n}(a'_{k},a''_{k},b'_{k},b''_{k},\locs_i,\ell_{i})} \para
    \prod^{m}_{i=1} \ell_i \near \ell \right) \text{where}\; 0 < k < n
\end{align*}

\begin{theorem}
For every $n \geq 0$, for every set of distinct names
$\set{a'_0,a''_0,b'_0,b''_0}$ and set of locations $\locs$ we have
\[ D^{0,n}(a'_0,a''_0.b'_0,b''_0,\locs,\ell) \wbisdot
\flatten{D^{0,n}(a'_0,a''_0,b'_0,b''_0,\locs,\ell),\ell} \]
\end{theorem}

\section{A type system} \label{sec:typesystem}

We can formulate a type discipline that distinguishes between broadcast and
collection and assigns bounds to these.

Types are given by the formation rules
\[ T ::= \btype{T}{C} \mid \btype{T}{B} \mid T^k \mid T_1 \times T_2
\cdots \times T_k \mid T_1 \ra T_2 \mid \textsf{Loc} \]
We now explain the intent behind these formation rules.

\begin{itemize}
\item The type $\btype{T}{C}$ is the type of a channel than can collect
messages of type $T$, and the type $\btype{T}{B}$ is the type of a channel
that can broadcast a message of type $T$.

\item A multiset with $k$ elements of type $T$ will have type $T^k$. 

\item A tuple with $k$ components of types $T_1, \ldots, T_k$, respectively,
will have type $T_1 \times \cdots \times T_k$ 

\item A constructor $f$ will have type $T_1 \ra T_2$, and a multiset $g$ will have
type $T^k \ra T$. If $g$ has type $T^k \ra T$, $g$ produces an element
of type $T$ from a multiset of $k$ terms, each having type $T$. 

\item Locations have type $\textsf{Loc}$.
\end{itemize}

A type environment $\Tenv$ is a finite function $\Tenv : \textbf{Names}
\cup \textbf{Sels} \cup \textbf{Cons} \rightarrow \textbf{Types}$ that
assigns types to all names, term constructors and multiset selectors. 

Type judgments are relative to a type environment and a set of agent variables $\Theta$ and
are of the form
\begin{align*}
  \Tenv \vdash M : T & \; \text{for terms} \\
  \Tenv \vdash^{\Theta} P & \; \text{for processes} \\
  \Tenv \vdash^{\Theta} N & \; \text{for networks}
\end{align*}

When the set is the same in all the judgements of a type rule, we omit
it throughout the rule.

\begin{table}[h]
\begin{center}
  \begin{tabular}{llll}
    \runa{Var} & $\Tenv \vdash  x : T$ if $\Tenv(x) = T$ &
\runa{SetVar} & $\Tenv \vdash S : T^k$ if $\Tenv(S) = T^k$ \\
\runa{Multiset} & \infrule{\Tenv \vdash M_i : T \quad 1 \leq i \leq k}{\Tenv \vdash
  \set{M_1, \ldots, M_k} : T^k} &
\runa{Select} & \infrule{E : T^k \quad \Tenv(g) = T^k \ra T}{\Tenv
  \vdash g(E) : T} \\
\runa{Cons} & \infrule{\Tenv \vdash M : T_1 \quad \Tenv(f) = T_1 \ra T_2}{\Tenv \vdash
  f(M) : T_2}
  \end{tabular}
\end{center}
  \caption{Type rules for message terms}
\end{table}

\begin{table}[h]
  \centering
  \begin{tabular}[c]{lp{5cm}ll}
    \runa{B-Inp} & \infrule{\Tenv, \vec{x} : \vec{T} \vdash M : T_1
      \quad \Tenv(a) = \btype{T_1}{B} \quad \Tenv, \vec{x} : \vec{T}
      \vdash P_1}{\Tenv \vdash \binp{a}{\patt{\vec{x}M}}P_1} \\
   \runa{C-Inp} & \infrule{\Tenv, \vec{x} : \vec{T} \vdash M : T_1
      \quad \Tenv(a) = \btype{T_1}{C} \quad \Tenv, \vec{x} : \vec{T}
      \vdash P_1}{\Tenv \vdash \binp{a}{\cpatt{\vec{x}M.S}}P_1} \\
\runa{Outp} & \infrule{\Tenv(a) = \btype{T}{d} \quad \mathsf{d} \in
  \set{\mathsf{B},\mathsf{C}} \quad \Tenv \vdash M : T_1}{\Tenv \vdash
  \outp{a}{M}P_1} \\
\runa{Nil} & $\Tenv \vdash \Nil$ &
\runa{Par} & \infrule{\Tenv \vdash P_1 \quad \Tenv \vdash P_2}{\Tenv \vdash P_1 \para
  P_2} \\
\runa{Sum} & \infrule{\Tenv \vdash P_1 \quad \Tenv \vdash P_2}{\Tenv \vdash P_1 +
  P_2} &
\runa{Match} & \infrule{\Tenv \vdash M_1 : T \quad \Tenv \vdash M_2 :
  T \quad \Tenv \vdash P_1}{\Tenv \vdash [M_1 = M_2]P_1} \\
\runa{Mismatch} & \infrule{\Tenv \vdash M_1 : T \quad \Tenv \vdash M_2 :
  T \quad \Tenv \vdash P_1}{\Tenv \vdash [M_1 \neq M_2]P_1} &
\runa{New} & \infrule{\Tenv, x : T \vdash P_1}{\Tenv \vdash \new{x :
    \beta_n}P_1} \\
\runa{Agent} & \infrule{\Tenv \vdash M_i : T_i \quad 1 \leq i \leq |\vec{M}|
  \quad A \in \Theta}{\Tenv \vdash^{\Theta} A(\vec{M})} \\
\runa{Call} & \infrule{\Tenv, \vec{x} : \vec{T} \vdash^{\Theta \cup
    \set{A}} \vdash P \quad \Tenv \vdash M_i : T_i \quad 1 \leq i \leq
  |\vec{M}| \quad A(\vec{x}) \defeq P}{\Tenv \vdash^{\Theta} A(M)}
  \end{tabular}
  \caption{Type rules for processes}
  \label{tab:types-processes}
\end{table}

\begin{table}[h]
  \centering
  \begin{tabular}[c]{ll}
    \runa{Par-N} & \infrule{\Tenv \vdash N_1 \quad \Tenv \vdash N_2}{\Tenv \vdash N_1 \para
  N_2} \\[3mm]
\runa{New} & \infrule{\Tenv, x : T \vdash N_1}{\Tenv \vdash \new{x :
    \beta_n}N_1} \\
\runa{Loc} & \infrule{\Tenv \vdash P \quad \Tenv(\ell) =
  \textsf{Loc}}{\Tenv \vdash \locate{P}} \\
\runa{Near} & \infrule{\Tenv(l) = \textsf{Loc} \quad \Tenv(m) =
  \textsf{Loc}}{\Tenv \vdash l \near m} \\
\runa{Not-Near} & \infrule{\Tenv(l) = \textsf{Loc} \quad \Tenv(m) =
  \textsf{Loc}}{\Tenv \vdash l \off m} \\
  \end{tabular}
  \caption{Type rules for networks}
  \label{tab:types-networks}
\end{table}

The following result, which guarantees that a well-typed network stays
well-typed under reductions, is easily established.

\begin{theorem}
If $\Tenv \vdash N$ and $N \ra N'$ then $\Tenv \vdash N'$
\end{theorem}
\begin{proof}
Induction in the reduction rules. A crucial lemma needed in the proof is that whenever
$N_1 \equiv N_2$, then $\Tenv \vdash N_1$ if and only if $\Tenv \vdash
N_2$. 
\end{proof}

\section{The expressive power of BBC} \label{sec:expressiveness}

There is no compositional encoding of \bbc into the $\pi$-calculus. In
the presence of broadcast this follows from the negative result due to
Ene and Muntean \cite{Ene:2001:BCC:645609.662485}. However, the result
also holds if we only consider collection. The following criteria for
compositionality are due to Palamidessi \cite{palamidessi} and Ene and
Muntean \cite{Ene:2001:BCC:645609.662485}.

\begin{definition}
An encoding $\sem{\quad}$ is \emph{compositional} if 
\begin{enumerate}
\item $\sem{N_1 \para N_2} = \sem{N_1} \para \sem{N_2}$.

\item For any substitution $\sigma$ we have $\sem{N \sigma} =
  \sem{N}\sigma$

\item If $N \ra^{*} N'$ then $\sem{N} \ra^{*} \sem{N'}$.

\item If $\sem{N} \ra^* M$ then $\sem{N} \ra^* M \ra^*
  \sem{N'}$ for some $N'$.

\end{enumerate}
\end{definition}

A central notion in the study of process calculi is that of an
\emph{electoral system}. This is a network in which the participants
perform a computation that elects a unique leader.  Following
\cite{palamidessi,Ene:2001:BCC:645609.662485}, our definition assumes
a network in which the names used are the `natural names' that
represent the identity of the $n$ processes in the network.

\begin{definition}[Electoral system]
  A network $N = P_1 \para \cdots \para P_n$ is an \emph{electoral system} if for every computation
  $C$, there exists an extension $C'$ of $C$, and an index $k \in
  \set{1,\ldots,n}$ such that for every $i \in \set{1,\ldots,n}$ the
  projection $C'_i$ contains exactly one output action of the form $k$
  and any trace of a $P_i$ may contain at most one action of the form
  $\overline{l}$ with $l \in \set{1,\ldots,n}$. 
\end{definition}
We now describe such a system in \bbc. The system uses a common
channel $a$ whose bound is $n$, where $n$ is the number of
principals. The idea is to collect names until $n$ sets of names have
been collected. The first to do so sends out a success announcement to
everyone in the form of a name $\textsf{chosen}(m)$; the collection
function $g(S)$ is defined to select the name with the minimal index
among the names of $S$ unless $S$ contains one or more occurrences of
a name of the form $\textsf{chosen}(m')$ for some $m$. Thus, the first
component to receive sufficiently many names from all other components
chooses the leader.

Our network is of the form 
\[ N = \prod^{n}_{i=1} \locate{l_i}{P_i} \para \prod^{n}_{i=1}
\prod^{m}_{j=1, j\neq i} l_i \near l_j \]
We define the $n$ components as follows:
\begin{align*}
 P_i \defeq 
& \outp{a}{i}\binp{a}{\cpatt{x.S_1}} \ldots \binp{a}{\cpatt{x.S_k}}
 \outp{a}{\textsf{chosen}(g(\set{\set{g(S_1), \ldots g(S_k)}}))}P 
\end{align*}
The selection function is defined by
\[ g(S) = \begin{cases}  \min_{x \in S} x & \text{if}\; S \subseteq
  \set{1,\ldots,n} \\
\textsf{chosen}(n) & \text{if}\; \textsf{chosen}(n) \in S 
\end{cases}
\]
As there is no such electoral system for the
$\pi$-calculus \cite{palamidessi}, we can use the same proof as
that given by Ene and Muntean to conclude that there can be no
compositional encoding of \bbc in the $\pi$-calculus. 

\section{Directions for further work}

In this paper we have presented the \bbc calculus, which is a
distributed process calculus that generalizes the $\pi$-calculus with
notions of channels with bounded forms of broadcast and collection and
an explicit notion of connectivity. 

The present work only considers barbed bisimulation; it is well-known
that this relation is not preserved by parallel composition and that
the notion of barbed congruence does not lend itself well to
coinductive proof techniques. Further work includes finding
a labelled characterization of barbed congruence; this requires a
labelled transition semantics in the spirit of \cite{Ene:2001:BCC:645609.662485}.

At present, we informally distinguish between channels for broadcast
and collection. However, if we think of channels as radio frequencies,
one would sometimes like the same frequency to be used for both
broadcast and collection. An interesting direction of work is
therefore to develop a type system that will allow the same channel to
be used non-uniformly in both modes and according to a protocol. This
is a variant of binary session types \cite{DBLP:conf/esop/HondaVK98}
that has been adapted to a setting with broadcast communication by
\cite{kouzapas2014}. A further step in this direction is to study
multiparty session types \cite{hondaPOPL} and how the projection to
local session types can be defined in the setting of \bbc.

\bibliographystyle{plain}

\end{document}